%% file: ISITPaper2.tex
\documentclass[conference]{IEEEtran}

\usepackage{amsmath}
\usepackage{color}
\usepackage{amssymb}
\usepackage{algorithm}
\usepackage{amsthm}
\usepackage{graphicx}
\usepackage{subfigure}
\usepackage{enumerate}

\newtheorem{thm}{Theorem}[section]
\newtheorem{lem}{Lemma}[section]

\newcommand{\x}{\mathbf{x}}
\newcommand{\y}{\mathbf{y}}
\newcommand{\f}{\mathbf{f}}

\newcommand{\beq}{\begin{equation}}
\newcommand{\eeq}{\end{equation}}  
\newcommand{\eea}{\end{align}}
\newcommand{\bea}{\begin{align}}

\newcommand{\li}{\left<}
\newcommand{\ri}{\right>}

\newcommand{\ab}{\mathbf{a}}

\newcommand{\bb}{{\mathbf{b}}}
\newcommand{\cb}{{\mathbf{c}}}

\newcommand{\Rb}{\mathbf{R}}

\newcommand{\ib}{{\mathbf{i}}}

\newcommand{\g}{\mathbf{g}}

\newcommand{\F}{\mathbf{F}}
\newcommand{\X}{\mathbf{X}}

\newcommand{\Pro}{{\mathbb{P}}}
\newcommand{\s}{{\mathbf{s}}}
\newcommand{\z}{{\mathbf{z}}}
\newcommand{\A}{{\mathbf{A}}}

\newcommand{\V}{{\mathbf{V}}}
\newcommand{\M}{\mathbf{M}}

\newcommand{\order}[1]{\mathcal{O}\left(#1\right)}

\begin{document}
\title{Sparse Phase Retrieval: Convex Algorithms and Limitations}
\author{\begin{tabular}[t]{c@{\extracolsep{5em}}c@{\extracolsep{5em}}c} 
Kishore Jaganathan & Samet Oymak  & Babak Hassibi \end{tabular}\\
 \\
        Department of Electrical Engineering \\
       California Institute of Technology \\
       Pasadena, CA~~91125
}
\date{}
\maketitle

\begin{abstract}
We consider the problem of recovering signals from their power spectral density. This is a classical problem referred to in literature as the phase retrieval problem, and is of paramount importance in many fields of applied sciences. In general, additional prior information about the signal is required to guarantee unique recovery as the mapping from signals to power spectral density is not one-to-one. In this paper, we assume that the underlying signals are sparse. Recently, semidefinite programming (SDP) based approaches were explored by various researchers. Simulations of these algorithms strongly suggest that signals upto $o(\sqrt{n})$ sparsity can be recovered by this technique. In this work, we develop a tractable algorithm based on reweighted $l_1$-minimization that recovers a sparse signal from its power spectral density for significantly higher sparsities, which is unprecedented.

We discuss the square-root bottleneck of the existing convex algorithms and show that a $k$-sparse signal can be efficiently recovered using $\order{k^2logn}$ phaseless Fourier measurements. We also show that a $k$-sparse signal can be recovered using only $\order{k log n}$ phaseless measurements if we are allowed to design the measurement matrices.

\end{abstract}

\let\thefootnote\relax\footnotetext{This work was supported in part by the National Science Foundation under grants CCF-0729203, CNS-0932428 and CCF-1018927, by the Office of Naval Research under the MURI grant N00014-08-1-0747, and by Caltech's Lee Center for Advanced Networking.}


\section{Introduction}
In many practical measurement systems, the power spectral density of the signal, i.e. the magnitude square of the Fourier transform, is the measurable quantity. Phase information of the Fourier transform is completely lost, because of which signal recovery is difficult. This problem occurs in many areas of engineering and applied physics, including X-ray crystallography \cite{millane}, astronomical imaging \cite{dainty},  microscopy, optics \cite{walther}, blind channel estimation and so on. 

Recovering a signal from its Fourier transform magnitude, or equivalently its autocorrelation, is known as phase retrieval. The mapping from signals to their Fourier transform magnitude is not one-to-one, and hence unique recovery is not possible in general. Additional measurements or prior information about the signal is required in order to uniquely recover the underlying signal. Constraints on the signal's values like non-negativity, bounds on the signal's support (locations where the value is non-zero), and more recently sparsity \cite{vetterli}, are commonly used as prior information. 

Considerable amount of research has been done over the last few decades (\cite{patt1, patt2}) and a wide range of heuristics have been proposed (see \cite{gerchberg}), a comprehensive survey of which can be found in \cite{fienup}. \cite{bauschke} provides a theoretical framework to understand the heuristics, which are in essence an alternating projection between a convex set and a non-convex set. Such methods often converge to a local minimum, hence drastically reducing the chances of successful signal recovery. 

Recently, the phase retrieval problem was recast as a semi-definite programming problem (see \cite{candespr,eldar} and \cite{kishore}). In \cite{candespr}, additional measurements with different illuminations, which is possible in an optical setup, are used to make unique recovery feasible. In \cite{eldar, kishore} and \cite{GESPAR}, the underlying signals are assumed to be sparse, which is a reasonable assumption in applications like X-ray crystallography, microscopy and astronomical imaging. 
 
Numerical simulations of the existing techniques based on SDP strongly suggest that signals upto $o(\sqrt{n})$ sparsity can be recovered with an arbitrarily high probability. This behavior for the phase retrieval problem was rigorously explained in \cite{kishore3, kishore2}. \cite{Samet, LiVor} consider the "generalized" phase retrieval problem and observe a similar behavior. In this work, we develop an algorithm based on the idea of reweighted $l_1$ minimization to solve the phase retrieval problem for significantly higher sparsities. We also provide certain theoretical guarantees and discuss the limitations of the SDP based techniques, and develop a combinatorial technique which requires far fewer measurements to guarantee recovery.

The remainder of the paper is organized as follows. In Section 2, we formulate the phase retrieval problem and recast it as a SDP problem. In Section 3, we discuss the limitations of the existing SDP-based techniques and develop an algorithm based on reweighted $l_1$ minimization in Section 4. In Section 5, we develop a measurement system using a combinatorial approach which requires far fewer measurements to guarantee recovery. Section 6 presents the results of the numerical simulations and concludes the paper.

\section{Problem Formulation}

Suppose $\x = (x_0, x_1, .... x_{n-1})$ is a real-valued discrete-time signal of length $n$ and sparsity $k$, where sparsity is defined as the number of non-zero entries.  Let $\y = (y_0,y_1,...y_{n-1})$ be its Fourier transform, i.e., 
\begin{equation}
\y=\F\x
\end{equation} 
where $\F$ is the $n \times n$ Discrete Fourier Transform (DFT) matrix. The phase retrieval problem can be mathematically stated as
\begin{align}
\label{PR}
\nonumber &\textrm{find} \quad \quad \quad  \quad \quad \x \\
& \textrm{subject to   } \quad \quad  |\y|=|\F\x| 
\end{align}

Since magnitude square of Fourier transform and autocorrelation are Fourier pairs, the phase retrieval problem can be reformulated as recovering signals from their autocorrelation, i.e., 
\begin{align}
& \textrm{find} \quad \quad  \quad \quad\hspace{5pt}  {\x} \\
\nonumber & \textrm{subject to}  \quad \quad  a_i =\sum_j{x_j x_{j+i}}  \quad  \quad  0 \leq i \leq n-1
\end{align}
where $\mathbf{a}=(a_0, a_1, ...., a_{n-1})$ is the autocorrelation of $\x$.

Observe that the operations of  time-shift, flipping and global sign-change do not affect the autocorrelation, because of which there is a trivial ambiguity. The signals resulting from these operations are considered equivalent, and in all the applications it is considered good enough if any equivalent signal is recovered.

The problem (\ref{PR}) is hard to solve due to non-convex constraints. We can relax the constraints into a set of convex constraints by embedding the problem in a higher dimensional space, a technique popularly known as lifting.  Note that (\ref{PR}) contains $n$ constraints of the form $|y_i|=|\f_i^T\x|$, where $\f_i$ is the $i^{th}$  column of $\F$. Squaring both sides, the constraints can be rewritten as $|y_i|^2=|\x^T\f_i^*\f_i^T\x|$. Suppose we define $\X=\x\x^T$, the problem can be recast in terms of $\X$ as
\begin{align}
\label{NCPR1}
\nonumber &\textrm{find} & &  \X \\
\nonumber & \textrm{subject to   } &  & |\y_i|^2=trace(\M_i\X) ~ ~ 0 \leq i \leq n-1\\
 & & & rank(\X)=1 \quad \& \quad \X \succcurlyeq 0
\end{align}
where $\M_i=\f_i^*\f_i^T$. 

In terms of the autocorrelation $\mathbf{a}$, the lifted problem can be formulated as

\begin{align}
\nonumber &\textrm{find} & &  \X \\
\nonumber & \textrm{subject to   } &  & \sum_{j=1}^n \X_{j,j+i}=a_i ~ ~ 0\leq i\leq n-1\\
 & & & rank(\X)=1 \quad \& \quad \X \succcurlyeq 0
\end{align}

\section{Background}

The program (\ref{NCPR1}) can be reformulated as a rank minimization problem as follows
\begin{align}
\label{NCPR2}
\nonumber &\textrm{minimize} & &  rank(\X) \\
\nonumber & \textrm{subject to   } &  & |\y_i|^2=trace(\M_i\X) ~ ~ 0 \leq i \leq n-1\\
& & & X \succcurlyeq 0
\end{align}
This is a non-convex problem as rank is a non-convex function. It has been shown in \cite{fazel1} that trace minimization is the tightest convex relaxation of rank minimization for positive semidefinite matrices. This relaxation is not useful in the phase retrieval setup as $trace(\X)$ corresponds to the energy of the signal $\x$, which is fixed by the magnitude of the Fourier transform. \cite{fazel2} proposes log-determinant function as a surrogate for rank in such problems, i.e.,
\begin{align}
\label{logdet}
&\textrm{minimize} \quad \quad  \textrm{log det }\mathbf{(X+\epsilon I)}  \\
& \textrm{subject to   } \quad  ~  |\y_i|^2=trace(\M_i\X) ~ ~ 0 \leq i \leq n-1\\
&  \quad \quad \quad \quad \quad ~~ X \succcurlyeq 0
\end{align}
This heuristic tries to minimize a concave function in a convex domain, which can be done using gradient descent approach. This method was explored for the phase retrieval setup in \cite{candespr, eldar}. Simulations suggest that the algorithm converges to a rank 1 solution with high probability if the underlying signal is $o(\sqrt{n})$ sparse.

In \cite{kishore}, we explored a two-stage recovery process to provably solve (\ref{PR}). In the first stage, we use information about the support of the autocorrelation to recover the support of the signal (see \cite{kishore3}). In the second stage, we solve the SDP with known support \cite{kishore2}. It was empirically observed and theoretically shown that signals were recovered with arbitrarily high probability if the sparsity was $o(\sqrt{n})$. However, if the support information was available by other means, it was observed that the program recovered signals up to roughly $o(n)$ sparsity.

\section{Recovery Algorithm}

In this section, we develop an iterative algorithm based on reweighted $l_1$-minimization to solve the phase retrieval problem outside the $o(\sqrt{n})$ sparsity regime. 

A two-stage approach like \cite{kishore} wouldn't work as the support of the autocorrelation becomes full if the signal has sparsity greater than $O(\sqrt{nlog(n)})$. Trace minimization in the phase retrieval setup has two issues: trivial ambiguities have same objective function, trace is fixed because of which we will be solving a feasibility problem only. Weighted $l_1$ minimization (\ref{WL1}) intuitively overcomes these issues and promotes sparse solutions.

\begin{align}
\label{WL1}
\nonumber &\textrm{minimize} & &  trace(\V|\X|) \\
\nonumber & \textrm{subject to   } &  & |\y_i|^2=trace(\M_i\X) ~ ~ 0 \leq i \leq n-1\\
& & & X \succcurlyeq 0
\end{align}
where $\V$ is a weight matrix, which can be designed to promote the necessary structure in the solution.

Simulations suggest that (\ref{WL1}) has rank 1 solutions with high probability if the sparsity is $o(\sqrt{n})$ and $\V$ is chosen randomly, but fails outside the $o(\sqrt{n})$ region. However, the largest eigenvalue turns out to be considerably stronger than the other eigenvalues, and the eigenvector corresponding to it happens to contain a lot of information about the support of the signal even significantly outside the $o(\sqrt{n})$ region, which is not very surprising.

This strongly suggests the possibility of an iterative algorithm, which at every iteration also knows "a lot"  about where the signal's non-zero entries can be. Algorithm 1 uses this information by doing a re-weighted minimization at every iteration. The weights corresponding to prospective support locations are set to zero to encourage the signal to choose those locations in the next iteration, and the weights outside this region is chosen to be positive. 

\begin{algorithm}
\label{algo1}
\caption{Phase Retrieval Algorithm}
\textbf{Input:} The magnitude of the Fourier transform $|\y|$, maximum number of iterations \\
\textbf{Output:} The underlying signal $\x$ \\
\begin{enumerate}[1.]
\item Initialize $\V$ by by choosing its entries from $[0,1]$ uniformly at random 
\item Solve the optimization problem
\begin{align}
\nonumber &\textrm{minimize}&  & trace(\V|\X|) \\
\nonumber & \textrm{subject to   }  & & |\y_i|^2=trace(\M_i\X) ~ ~ 0 \leq i \leq n-1\\
& & & X \succcurlyeq 0
\end{align}
\item If $rank(\X)=1$, return $\X$, else calculate $\X_1=\x_1\x_1^T$, where $\X_1$ is the best rank-1 approximation of $\X$
\item Update $\V$ as follows: $\V_{ij}=0$ if $|\x_i|$ and $|\x_j|$ are greater than a certain threshold, choose the remaining entries from $[0,1]$ uniformly at random
\item Iterate until convergence or maximum number of iterations
\item Calculate $\X_1$ and return $\x_1$
\end{enumerate}
\end{algorithm}

\input{ClassicalPSD}

\input{PreviousWork}

\input{TractableGood}

\section{Numerical Simulations}
\label{numerical}

In order to demonstrate the performance of Algorithm 1, numerical simulations were performed for different values of signal length $n$ and sparsities $k$. 

For a given $n$ and sparsity $k$, the $k$ support locations were chosen from the $n$ locations uniformly at random. The signal values at the support locations were chosen from [0,1] uniformly at random. The Fourier transform magnitudes of the signal were computed and provided as input to Algorithm 1. If the output signal matched the input signal, it was counted as a success, else it was counted as a failure. Two other semidefinite program based algorithms (CandesPR: \cite{candespr}, HassibiPR: \cite{kishore}) and the traditional alternating projection algorithm (GS: \cite{gerchberg}) were used for comparison purposes.

The results of the numerical simulations are shown in Figure 1. The probability of successful recovery is plotted against various sparsities for $n=32$ and $n=64$. $100$ simulations were performed for each sparsity to get the average success probability, which was calculated as the percentage of successful recovery.  The existing SDP based algorithms start fading at around $o(\sqrt{n})$ sparsity, the figure clearly demonstrates the superiority of using reweighted minimization outside the $o(\sqrt{n})$ sparsity region. 

\begin{figure}[h]
\begin{center}
\includegraphics[scale=0.5]{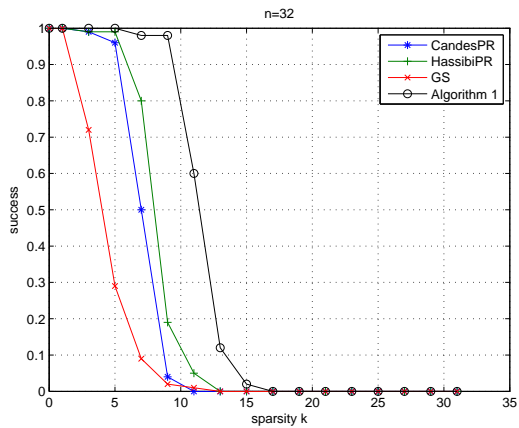}  
\end{center}
\caption{Success rate of recovery for $n=32$ and various sparsities}
\label{fig1}
\end{figure}

\begin{figure}[h]
\begin{center}
\includegraphics[scale=0.5]{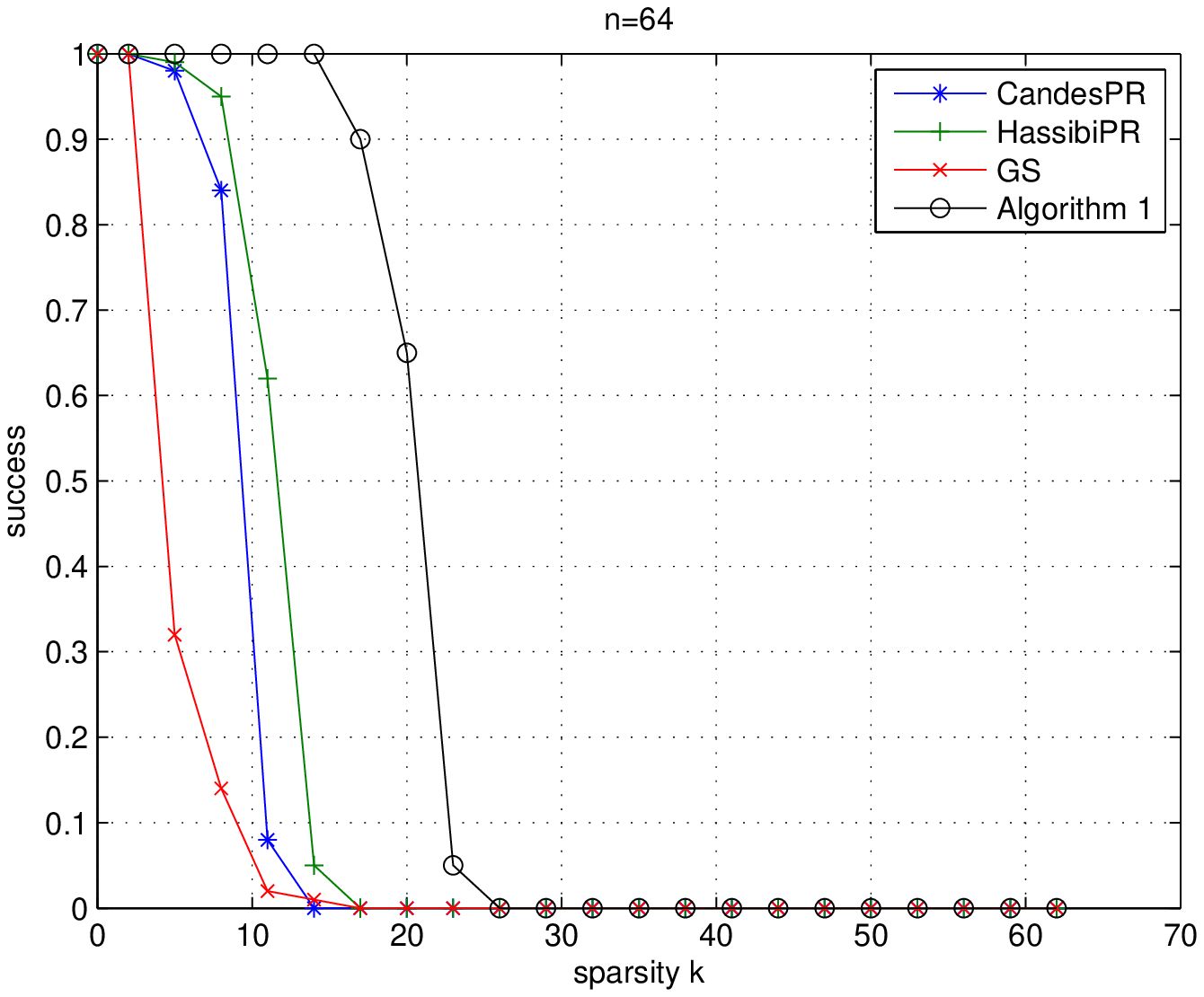}  
\end{center}
\caption{Success rate of recovery for $n=64$ and various sparsities}
\label{fig2}
\end{figure}

\newpage

\appendix
{\bf{Proof of Theorem \ref{phaselift}}}
\begin{proof} The proof provides a tractable recovery algorithm and consists of two main steps.
\begin{enumerate}[]
\item {\bf{Support recovery:}} Denote the support of $\x$ by $S\subseteq[n]$ and the support of $\z_i$ by $S_i$ and let $\cb_i=\z_i\cdot \ab_i$. Now, consider the inner product $\cb_i^*\x$. Since nonzero entries of $\cb_i$ have continuous distribution if $S\cap S_i\neq \emptyset$, $|\cb_i^*\x|\neq 0$ almost surely. Hence whenever a measurement $|\cb_i^*\x|^2=0$ we can deduce that $S\cap S_i=\emptyset$. Let:
\beq
I=\{1\leq i\leq m\big|S_i\cap S=\emptyset\}
\eeq
Clearly, $\{S_i\}$'s are i.i.d. supports and for each $i$ we have:
\beq
\Pro(S\cap S_i=\emptyset)=(1-1/k)^k\geq \frac{1}{4}\label{14eq}
\eeq
Following \eqref{14eq}, by the law of large numbers, w.h.p. $|I|\geq m/8$. Conditioned on $|I|\geq m/8$, the probability that $j\in\bar{S}$ is not contained in $\bigcup_{i\in I}S_i$ is at most $(1-1/k)^{|I|}\leq (1-1/k)^{m/8}$.

Assuming $m\geq 8k\log n$ and using a union bound:
\begin{align}
\Pro(\bar{S}\not\subseteq\bigcup_{i\in I}S_i)&\leq \exp(-\frac{m}{8k})\leq  n^{-1}
\end{align}
which will approach $0$. Hence, the exact support of the signal can be found by simply taking the union of sets $S_i$ satisfying $\cb_i^*\x=0$ and then complementing it ($O(mn)$ time). Next, with the knowledge of support, we proceed with the recovery of $\x$ up to an overall phase ambiguity.
\item {\bf{Signal recovery:}}
Recovery of the signal given its support will be performed in two steps. We first show that magnitudes of nonzero entries of $\x$ can be found.
\begin{enumerate}[]
\item {\bf{Recovering magnitudes:}} Assume $S_i\cap S$ is a singleton $j\in S$. Since we already have the knowledge of $S$ from previous part, we can immediately deduce: $|x_j|^2=\frac{|\li\cb_i,\x\ri|^2}{|c_{i,j}|^2}$.
Then, we simply need to ensure that w.h.p., for all $j\in S$ there exists $1\leq i\leq m$ satisfying $S\cap S_i=j$. Probability of this not happening for a fixed $j\in S$ is:
\beq
(1-\frac{1}{k}(1-\frac{1}{k})^{k-1})^m\leq \exp(-\frac{m}{4k}).
\eeq
After union bounding over all $j\in S$, whenever $m\geq 8k\log n$, all $|x_j|^2$ can be found with probability at least $1-k\exp(-\frac{m}{4k})\geq 1-n^{-1}$.

\item {\bf{Recovering the relative phases:}} Next, we consider the measurements satisfying $|S_i\cap S|=2$, $1\leq i\leq m$. For fixed $i$ we have:
\beq
\Pro(|S_i\cap S|=2)={k\choose 2}\frac{1}{k^2}(1-\frac{1}{k})^{k-2}\geq \frac{1}{8}.
\eeq
Overall, w.h.p. there are $m/10$ measurements satisfying $|S_i\cap S|=2$. Let:
\beq I=\{1\leq i\leq m |\bigcup |S_i\cap S|=2\}
\eeq
Next, form the $k$ vertex graph obtained by connecting the nodes $j,l$ whenever $\{j,l\}=S_i\cap S$ for some $i\in I$ ($O(mn)$ time). Observe that, each $i\in I$ picks an edge in this graph uniformly at random. Overall, we perform at least $\frac{m}{10}$ picks with replacement. This graph is connected with high probability whenever the chance of an edge being picked is more than $\frac{2\log k}{k}$, \cite{Erdos}. This happens when $m\geq ck\log k$ for $c\geq 10$ since:
\begin{align}
\Pro(\text{edge is picked})&\geq1- (1-\frac{1}{{k\choose 2}})^{m/10}\\
&\geq 1-\exp(-\frac{m}{5k(k-1)})\\
&\geq 1-\exp(-\frac{c\log k}{5k})\approx\frac{c\log k}{5k}
\end{align}
Now, w.h.p. the graph is connected. Find a spanning tree $T$ in this graph, which can be done in $O(k^2)$ time, \cite{Fredman}. Set the phase of the initial node to $0$. Then, the rest of the phases of the nodes are uniquely determined as follows.
\item {\bf{Recovering relative phases:}} To begin, consider an edge of $T$ between nodes $j,l$ where $\{j,l\}=S_i\cap S$ for some $i$. Fix $x_j=|x_j|$ i.e. $0$ phase. From measurements:
\beq
|\li\ab_i\cdot \z_i,\x\ri|^2,|\li\bb_i\cdot \z_i,\x\ri|^2
\eeq
and with the knowledge of $|x_i|^2,|x_j|^2$, we can find: $\Re(z_jz_la_j^*a_lx_jx_l^*)$ and $\Re(z_jz_lb_j^*b_lx_jx_l^*)$ where $\Re(\cdot)$ returns the real part of a number. This information is equivalent to: $\Re(ax), \Re(bx)$ where $a=a_j^*a_l$, $b=b_j^*b_l$ have uniformly random phases and $x=x_jx_l^*$. We will argue that the phase of $x$ is uniquely determined when $\Re(ax), \Re(bx), a,b,|x|$ are known. Assume $a=\exp(\ib\theta_1)$, $b=\exp(\ib \theta_2)$ and $x=|x|\exp(\ib\theta)$. Then,
\begin{align}
&\Re(ax)=|x|\cos(\theta+\theta_1)\\
&\Re(bx)=|x|\cos(\theta+\theta_2)
\end{align}
gives us two linearly independent equations (almost surely) and two unknowns ($\sin \theta,\cos\theta$). Overall, this would yield $\theta$. Applying this over all edges of the tree recursively, we can find the exact phase differences between all neighboring nodes and since the graph is connected we can find the phase of any $j\in S$ by adding up the phase differences over all edges on the path between the initial node and $j$. This can be done for all nodes in $O(k^2)$ time via DFS. Overall, the original signal $\x$ can be found up to an overall phase ambiguity in $O(mn)$ time.
\end{enumerate}
\end{enumerate}
\end{proof}
\end{document}

%% file: ClassicalPSD.tex
\section{Phase Retrieval as a compression problem}

Theorem \ref{thm1} was proved for the phase retrieval problem in \cite{kishore3, kishore2}.
\begin{thm} \label{thm1}
Signals can be recovered from their power spectral densities up to time-shift, reversal and global sign, in polynomial time, with probability $1-\delta$ for any $\delta>0$  if
\begin{enumerate}
\item $k=o(\sqrt{n})$
\item $k$ entries are chosen uniformly at random
\item $n>n(\delta)$
\end{enumerate}
\end{thm}
We can instead fix the sparsity $k$ and consider the phase retrieval problem with partial power spectral density information (for example \cite{moravec}). In particular,  one might have access to only certain frequencies. To solve this problem, we will use some classical results in the compressed sensing literature.

\begin{thm}[Candes, Romberg, Tao, \cite{Candes}]\label{thm2} Let $\F$ be the $n\times n$ DFT matrix. Consider the random matrix $\A$ obtained by choosing $m$ rows of $\F$ uniformly at random (without replacement). If $\x$ is a $k$ sparse vector, $\x$ can be recovered from observations $\A\x$ with arbitrarily high probability if $m\geq O(k\log n)$ via the following $\ell_1$ minimization:
\begin{align}
\label{l1min}
& \min_{\hat\x}\quad \quad \quad \quad \|\hat\x\|_1 \\
& \nonumber \text{subject to} \quad \quad \A\hat\x=\A\x
\end{align}
\end{thm}

The following theorem gives a useful result for the recovery of a sparse signal from its partial power spectral density.

\begin{thm} \label{thm3}Let $\x$ be a $k$ sparse vector satisfying the conditions of Theorem \ref{thm1}. Suppose we observe its power spectral density at $m$ distinct frequencies chosen uniformly at random (without replacement). If $m\geq O(k^2\log n)$, $\x$ can be recovered from its partial power spectral density, in polynomial time, with arbitrarily high probability.
\end{thm}

\begin{proof}
If $\x$ is a $k$-sparse signal, its autocorrelation can have at most $k^2$ non-zero entries. Since power spectral density is the Fourier transform of the autocorrelation, using Theorem \ref{thm2}, whenever $m\geq O(k^2\log n)$, the autocorrelation of $\x$ can be recovered from partial power spectral density observations via $\ell_1$ minimization with arbitrarily high probability. Now that the autocorrelation is found, Theorem \ref{thm1} completes the proof.
\end{proof}

While Theorem \ref{thm3}, gives a tractable algorithm for the recovery of $\mathbf{a}$ from partial power spectral density observations $\s=\A\ab$, one can consider a more sophisticated approach . The program (\ref{l1min}) tries to solve for a $k^2$-sparse solution without using the extra information that the resulting signal should be a valid autocorrelation. We propose the following method which might be of interest for future directions.
\begin{align}
&\min_{\ab} \quad \quad \quad \quad \quad \|\ab\|_1 \label{abc2} \\
\nonumber &\text{subject to} \quad \quad \quad \A\ab=\s\\
\nonumber & \quad \quad \quad \quad \quad \quad \quad \sum_{j=1}^n \X_{j,j+1}=a_i\quad 0\leq i\leq n-1  \\
 & \quad \quad \quad \quad \quad \quad \quad rank(\X)=1 ~ ~ ~\& ~ ~ ~ \X\succeq 0 \label{abc}
\end{align}

Observe that the only nonconvex constraint in \eqref{abc2} is \eqref{abc}. We believe that solving \eqref{abc2} can substantially increase the performance and instead of $O(k^2\log n)$ measurements, just $O(k\log n)$ measurements might suffice which is similar to that of typical compressed recovery of a $k$ sparse vector. However, further investigation is required to relax \eqref{abc} in a useful manner as opposed to dropping it.

Overall, Theorem \ref{thm3} is subject to a $o(\sqrt{n})$-sparsity bottleneck since the full power spectral density corresponds to $n$ measurements. While we do not provide theoretical guarantees for Algorithm 1, when full power spectral density is available, our algorithm seemingly beats the $o(\sqrt{n})$-sparsity bottleneck. In section \ref{numerical}, we see that the recoverable sparsity is much higher than $o(\sqrt{n})$.

\subsection{Two-stage recovery}
In \cite{kishore}, we try to solve the sparse phase retrieval problem via a two-stage approach where the first step involves finding the support of the signal from the support of the autocorrelation. We will now argue that such an approach is inherently subject to the $o(\sqrt{n})$ bottleneck, as there is no way of finding the support of the signal when its sparsity is greater than $O(\sqrt{n})$.
\begin{lem} Suppose $\x$ is a $k$-sparse signal whose support is chosen uniformly at random, and whose nonzero entries are continuous i.i.d. random variables. Then, there exists a constant $c$ such that whenever $k\geq c\sqrt{nlog(n)}$, support of the autocorrelation is full with arbitrarily high probability.
\end{lem}
\begin{proof}
Without loss of generality, we can assume that each location belongs to the support of the signal with probability $\frac{k}{n}$ independently as the same proof will apply for $k$-sparse signals with standard modifications.

For a particular distance $d$, if no two non-zero entries in the signal are separated by a distance $d$, we can say that $d$ doesn't belong to the support of the autocorrelation. This probability can be bounded by $(1-k^2/n^2)^{n/2}$ which is upper bounded by $e^{-k^2/2n}$. Union bound tells us that the probability of the support of the autocorrelation not being full is less than $ne^{-k^2/2n}$, which goes to zero if $k\geq c\sqrt{nlog(n)}$ for sufficiently large $n$.
\end{proof}

%% file: PreviousWork.tex
\subsection{Relation to Gaussian Phase Retrieval}
Our results on partial power spectral density can be related to the "generalized" phase retrieval problem, where the observations are of the form $|\g_i^T\x|$ for i.i.d. complex standard normal vectors $\{\g_i\}_{i=1}^m$. While this problem is structurally similar to phase retrieval, it is considerably simpler as there are no trivial ambiguities like time-shift and flipping.

Assuming $\x$ is a sparse vector, \cite{Samet} and \cite{LiVor} analyze the following semidefinite program for the tractable recovery of $\x$ up to a global phase ambiguity.
\begin{align}
&\min_{\hat\X} ~ ~ ~ ~ ~ ~ ~ ~ ~ ~\|\hat\X\|_\star+\lambda \|\hat\X\|_1\label{sdp1}\\
\nonumber &\text{subject to} ~ ~ ~ \li\g_i\g_i^T,\hat\X\ri=\li\g_i\g_i^T,\x\x^T\ri ~ ~ ~ 1\leq i\leq m
\end{align}
Naturally, one would wish that the unique minimizer of \eqref{sdp1} to be $\x\x^T$ to be able to recover $\x$ up to phase ambiguity. \eqref{sdp1} tries to capture both sparsity and low rankness of the underlying matrix $\x\x^T$. Interestingly, both \cite{Samet} and \cite{LiVor} suggest that as long as $m\leq O(\min\{k^2,n\})$, recovery using (\ref{sdp1}) is not possible with very high probability for any choice of regularizer $\lambda$. To summarize, even if $k^2\ll n$, one would still need $\Omega(k^2)$ measurements for recovery, which indicates an $o(\sqrt{n})$ bottleneck for generic Gaussian measurements too.

On the other hand, \cite{Samet}, \cite{LiVor} and \cite{Eldar2} show that, if one is able to search for a low rank and sparse matrix, then, $\x\x^T$ can be recovered with only $k\log \frac{n}{k}$, measurements which illustrates a significant performance gap between tractable recovery and intractable combinatorial search. Interestingly, \cite{kishore} illustrates a similar gap for the phase retrieval problem with Fourier measurements.

Overall, we have seen that for the two important class of phaseless measurements (Fourier and Gaussian), the current theoretical guarantees for the existing algorithms are subject to a strong $o(\sqrt{n})$-sparsity bottleneck. A natural question is whether it is possible at all to do sparse phase retrieval in a tractable way with small number of measurements.

%% file: TractableGood.tex
\subsection{Combinatorial Approach}
To address this question, we consider the problem of recovering $\x$ from phaseless measurements by making use of a specific choice of measurements. In particular, we show that one can tractably recover a $k$-sparse signal from phaseless measurements by using only $O(k\log n)$ measurements with very high probability. This shows that phase retrieval with optimal number of measurements (up to a $\log n$ factor) is in fact possible.
\begin{thm}\label{phaselift} Suppose $\x$ is an arbitrary $k$-sparse vector where $k\geq 2$. Let $\{\z_i\}_{i=1}^m$ be i.i.d. vectors with i.i.d. entries distributed as:
\beq
\begin{cases}0~~~~~~~~~~~~\text{ with probability}~~1-1/k \\ \mathcal{N}(0,1) ~ ~~~~\textrm{with probability} ~~1/k\end{cases}\label{measdesign}
\eeq
Let $\{\ab_i,\bb_i\}_{i=1}^m$ be i.i.d. vectors with i.i.d. entries distributed as $\exp(\ib\theta)$ where $\theta$ is uniformly distributed in $[0,2\pi)$. Denote the function $\Rb^n\times\Rb^n\rightarrow\Rb^n$ that returns entrywise products of two vectors by $\cdot$. Assume, we observe the measurements:
\beq
|\li\ab_i\cdot \z_i,\x\ri|^2,|\li\bb_i\cdot \z_i,\x\ri|^2
\eeq
for $1\leq i\leq m$. 

Then, $\x\x^*$ can be recovered with high probability in $O(mn)$ time, whenever $m\geq ck\log n$ for some constant $c>0$.
\end{thm}
{\bf{Remark:}} In general, we don't need the knowledge of sparsity $k$ for the design of the measurement operator. This can be handled by introducing an extra factor of $\log n$ in the required number of measurements, i.e. $klog^2(n)$ measurements. The reason is we can take $ck\log n$ measurements designed for each of the sparsity levels of $k_i=2^i$ for $1\leq i\leq {\lceil{\log n}\rceil}$, i.e. use $k_i$ instead of $k$ in \eqref{measdesign}. Then, $k$ will lie between $k_i$ and $k_{i+1}$ for some $i$ and the same proof argument will work.\\

The proof of Theorem \ref{phaselift} can be found in the appendix. While our proposed recovery algorithm requires a perfectly sparse signal, we emphasize that, our intention with Theorem \ref{phaselift} is demonstrating the possibility of tractable recovery rather than coming up with a robust algorithm for artificially designed measurements \eqref{measdesign}.